\documentclass[11pt]{article}

\usepackage[margin=1in]{geometry}
\usepackage{amsmath,amssymb,amsthm}
\usepackage{mathtools}
\usepackage{hyperref}
\usepackage{enumitem}
\usepackage{authblk}
\usepackage{physics}
\usepackage{cleveref}
\usepackage{bm}
\usepackage{algorithm}
\usepackage{algorithmic}
\usepackage{color}

\title{Universality of Classically Trainable, Quantum-Deployed Boson-Sampling Generative Models}

\newcommand{\aqa}{$\langle aQa^L \rangle$ Applied Quantum Algorithms, Leiden University, The Netherlands}
\newcommand{\liacs}{LIACS, Leiden University, Niels Bohrweg 1, 2333 CA, Leiden, The Netherlands}
\newcommand{\ENS}{DIENS, \'Ecole Normale Supérieure, PSL University, CNRS, INRIA, 45 rue d'Ulm, 75005 Paris, France}
\newcommand{\Wigner}{HUN-REN Wigner Research Centre for Physics, Budapest, Hungary}
\newcommand{\ELTE}{E\"otv\"os Lor\'and University,  Budapest, Hungary}
\newcommand{\UH}{University of Helsinki, Yliopistonkatu 4 00100 Helsinki, Finland}
\newcommand{\ALGO}{Algorithmiq Ltd, Kanavakatu 3C 00160 Helsinki, Finland}

\author[1,2]{Andrii Kurkin}
\author[3]{Ulysse Chabaud}
\author[4,5]{Zolt\'an Kolarovszki}
\author[4,5]{Bence Bak\'o}
\author[6,7]{Zolt\'an Zimbor\'as}
\author[1,2]{Vedran Dunjko}
\affil[1]{\aqa}
\affil[2]{\liacs}
\affil[3]{\ENS}
\affil[4]{\Wigner}
\affil[5]{\ELTE}
\affil[6]{\UH}
\affil[7]{\ALGO}
\date{\today}

\theoremstyle{plain}
\newtheorem{theorem}{Theorem}

\newtheorem{lemma}{Lemma}

\newtheorem{cor}[theorem]{Corollary} 

\theoremstyle{definition}
\newtheorem{definition}{Definition}

\theoremstyle{remark}
\newtheorem{remark}{Remark}

\begin{document}
\maketitle

\begin{abstract}

Recent work on the instantaneous quantum polynomial-time (IQP) quantum-circuit Born machine (QCBM) highlights a promising paradigm for generative modeling: \emph{train classically, deploy quantumly}. In this setting, the training objective can be evaluated efficiently on a classical computer, while sampling from the resulting model may still be classically intractable. Furthermore, in the IQP-QCBM framework, extending the model family with ancillary qubits has been proven to yield universality. This paper asks whether similar results hold for linear-optical generative models. To this end, we introduce the \emph{Boson Sampling Born Machine} (BSBM). Our analysis retraces analogous steps as were found for IQP-QCBMs with twists. Using recent results that enable classical approximation of broad classes of expectation values in linear optics, we show that BSBMs can be trained classically for wide families of loss functions. Next, we argue that ``basic'' BSBMs are not universal generative models, and that universality can be achieved by expanding the model while preserving efficient classical training and sampling hardness. In our approach, we introduce and analyze the role of constant-function postprocessing, generalizing the construction for IQP-QCBMs, which under suitable conditions can lead to universality while preserving the hardness of classically simulating the models. We showcase a family of BSBMs, characterized by a single hyperparameter, that allows for a monotonic increase in expressivity toward universality while retaining the capacity to represent ostensibly hard distributions. Furthermore, we discuss the possible modalities for the efficient classical training, in the sense of efficient estimation of gradients of the loss function. 
\end{abstract}

\section{Introduction}
Quantum computing may enable new capabilities in machine learning, however showing a \emph{practical} quantum advantage, where near-term quantum models outperform classical ones on real tasks, still remains challenging as this requires real-world evaluation of models. Recent works in quantum generative modeling \cite{recio2025train, recio2025iqpopt, Kurkin2025IQP, bako2025fermionic} have introduced the \emph{train classically, deploy quantumly} paradigm: one optimizes model parameters using objectives that can be estimated efficiently on classical hardware, and then use a quantum device at deployment stage to generate samples from the trained model, potentially from distributions that remain hard to sample classically. In particular, instantaneous quantum polynomial-time (IQP) quantum-circuit Born machines (QCBM) were shown to be efficiently trainable at scale under the Maximum Mean Discrepancy (MMD) objective with a Gaussian kernel \cite{recio2025train}, which was later extended to a broader class of kernels in \cite{Kurkin2025IQP}. Moreover, it was shown that introducing hidden qubits renders the model universal \cite{Kurkin2025IQP}. In parallel, quantum generative modeling has also been explored in the photonic setting, with prior work on linear-optics-based generative modeling demonstrating proof-of-concept photonic hardware implementations \cite{sedrakyan2024photonic} and noise-mitigation strategies for photonic QCBMs \cite{salavrakos2025error}. In this paper, we ask whether an analogous picture to IQP-QCBMs can be realized in \emph{linear-optics}-based generative modeling.

We focus on boson sampling as the basis for our generative model, which we call the \emph{Boson Sampling Born Machine} (BSBM). Roughly speaking, a BSBM is a boson sampler with trainable optical-circuit parameters: measuring the output in the Fock basis yields samples from the model distribution. On the one hand, there is ample evidence that classical simulation of sampling from these models is intractable \cite{aaronson2011computational}. On the other hand, a stream of algorithmic advances has shown that many expectation values and related statistics of linear-optical circuits can be approximated efficiently on a classical computer to additive error \cite{gurvits2005complexity,aaronson2012generalizing,chabaud2021quantum,seron2024efficient,LimOh2025}. The analogous results of \cite{nest2009simulating} for IQP were key to establishing the classical trainability of IQP-QCBMs. Building on these results, we show that many practically relevant loss functions for BSBMs remain classically tractable to evaluate. The central question is whether we can combine this classical trainability with (i) meaningful expressive power (ideally approaching universality) and (ii) parameter regimes in which deployed sampling remains hard. We answer this question in the affirmative by providing explicit constructions.

This paper is organized as follows. \Cref{sec:preliminaries} introduces boson sampling generative models and discuss their non-universality. \Cref{sec:classical-training} presents efficient classical estimation algorithms for some training loss functions and discusses implications for training and gradient computation. 
\Cref{sec:universality} introduces a sequence of extended models, proves monotonic expressivity and universality in the large-mode limit, and explains how sampling hardness can persist in intermediate regimes. \Cref{sec:discussion} concludes with open problems related to postselection and adaptivity, and more general input states.

\section{Preliminaries}\label{sec:preliminaries}
This section fixes notation for linear-optical (boson sampling) distributions, introduces the Boson Sampling Born Machine (BSBM) used throughout the paper, and discusses its non-universality with respect to distribution approximation.

\subsection{Linear optics and boson sampling}
\label{sec:linearoptics}

We consider an $m$-mode passive linear-optical interferometer described by a unitary matrix
$U \in \mathcal{U}(m)$ acting on mode operators as
\begin{equation}
    b^\dagger_j = \sum_{i=1}^m U_{j i}\, a^\dagger_i,
    \qquad
    b_j = \sum_{i=1}^m U^*_{j i}\, a_i,
\end{equation}
where $a^\dagger_i$ ($a_i$) and $b^\dagger_j$ ($b_j$) are input and output creation (annihilation)
operators, respectively. The input state is typically a tensor product of Fock states. In the standard
boson sampling setting, one injects $k$ single photons into $k$ designated input modes, e.g.,
\begin{equation}
    \ket{\psi_{\mathrm{in}}} = \ket{1^k 0^{m-k}} \equiv \ket{1, \ldots, 1, 0, \ldots, 0},
\end{equation}
with one photon in each of the first $k$ modes and vacuum in the remaining modes. The circuit is measured nonadaptively in the photon-number (Fock) basis, producing an outcome $\bm{s}=(s_1,\ldots,s_m)\in \mathbb{N}^m$ with total photon number $\sum_{j=1}^m s_j = k$. For further background on linear optics and boson sampling see \cite{aaronson2011computational,gard2015introduction}.

A bosonic quantum computer is believed to be a non-universal model of quantum computation, while also widely believed to be classically hard to sample from the induced distribution. The work of Aaronson and Arkhipov showed that an efficient classical algorithm for (exact, and under plausible conjectures also approximate) boson sampling would imply unlikely complexity-theoretic collapses \cite{aaronson2011computational}. A common regime for complexity-theoretic hardness is the \emph{collision-free} (or anti-bunching, diluted) regime, where $m \gg k^2$ and the probability of observing any $s_j \ge 2$ is small. Conditioning
on collision-free outcomes yields the outcome space
\begin{equation}
    \Omega_m^k := \left\{\bm s\in \{0,1\}^m : \|\bm s\|_1 = k\right\}.
\end{equation}
For $\bm s\in\Omega_m^k$, define the occupied-mode set
$S(\bm{s}) := \{j\in[m] : s_j=1\}$, a $k$-subset of $[m]:=\{1,\ldots,m\}$. For the input
$\ket{1^k0^{m-k}}$, the collision-free output probabilities are given by
\begin{equation}
    q_U(\bm s)
    = \left| \mathrm{Per}\!\left(U_{S(s),[k]}\right) \right|^2,
    \label{eq:bs-permanent}
\end{equation}
where $U_{S(\bm s),[k]}$ is the $k\times k$ submatrix formed by selecting rows indexed by
$S(\bm s)$ and columns $[k]$. In the rest of this work, we focus on this collision-free regime and use boson sampling as the basis for our generative models.

\subsection{Boson Sampling Born Machine}

\label{sec:bsgm}
We parameterize a family of interferometers by the unitary describing the linear network $U\in \mathcal{U}(m)$.
Fixing an input state and measurement scheme induces a parametric family of distributions.
\begin{definition}[Boson Sampling Born Machine (BSBM)]\label{def:bsbm}
A Boson Sampling Born Machine (BSBM) is a generative model specified by: the dimension of the output distribution $n$ (here we focus on $n$-bit bitstrings), the number of optical modes used $m$, the linear optical network specifying the boson sampler described by the unitary of the optical network $U\in \mathcal{U}(m)$, and a fixed initial state of $k$ photons (we will assume $m \gg k^2$ to remain in the diluted regime).
Hereafter, the initial state is fixed to be the $k$-photon input state $\ket{\psi_{\mathrm{in}}}=\ket{1^k0^{m-k}}$, that is, one photon in each of the first $k$ modes and vacuum in the remaining $m-k$ modes.
By a photon number measurement of all output modes, postselecting on the no-collision case, it represents a family of distributions
\begin{equation}
    \mathcal{Q}(m, k):=\{q_{m,k,U}\}_{U \in \mathcal{U}(m)},
\end{equation}
with the the sample space  $\Omega_m^k:=\{\bm s\in\{0,1\}^m:\|\bm s\|_1=k\}$, and where $q_{m,k,U}$ is a probability distribution on $\Omega_m^k$ given by
\begin{equation}\label{eq: basic model induced probability}
    q_{m,k,U}(\bm s)
    =
    \bigl|\mathrm{Per}\bigl((U)_{S(\bm s),[k]}\bigr)\bigr|^2.
\end{equation}
Here, $S(\bm s):=\{j\in[m]: s_j=1\}$ denotes the corresponding $k$-element subset of occupied modes.

While $q_{m,k,U, f}$ above is defined as a distribution, depending on the context we will also use it to refer to the underlying BSBM realizing it.
\end{definition}

\begin{remark}[Specification of the model]
Each boson-sampler-based generative model is specified by the initial state, the optical network, and the final measurements. In this paper, the initial state to be a Fock state with only one photon occupying each of the first $k$ modes, and the measurement is fixed. Consequently each model is specified by the unitary specifying the linear network, implicitly the number of modes $m$, and the photon number $k$. Depending on the scenario we may change notation and explicitly characterize a model by the unitary, or in certain cases, mode number and photon number, but in all cases the full specification depends on the unitary and the photon number $k$.
\end{remark}

\begin{definition}[Universality of a generative model]
Let $\mathcal{X}$ be an output space, and let $\mathcal{P}(\mathcal{X})$ denote the set of all probability distributions on $\mathcal{X}$.
A generative model $\{q_U\}_{U \in \mathcal{U}()},\subseteq \mathcal{P}(\mathcal{X})$ is \emph{universal} with respect to a metric $d(\cdot,\cdot)$ on $\mathcal{P}(\mathcal{X})$ if for every target distribution $p\in\mathcal{P}(\mathcal{X})$ and every $\varepsilon>0$, there exists a parameter setting $U$ such that $d(p,q_U)\leq \varepsilon$.
\end{definition}

The BSBM model $\mathcal{Q}(m, k)$ is \emph{not} universal over $\{0,1\}^m$. A basic obstruction is support: the model assigns zero probability to every bitstring of Hamming weight $>k$ (in particular, to the all-ones string), hence it cannot represent arbitrary distributions on $\{0,1\}^m$. Even if one ignores this support restriction, a further limitation comes from degrees of freedom. An arbitrary probability distribution on $\{0,1\}^m$ has $2^m-1$ independent parameters, so universality over $\{0,1\}^m$ requires a model class with effectively exponentially many degrees of freedom in $m$ (a dimension-counting argument that can be made rigorous). While the boson sampling measurement outcome space can itself be exponentially large, scaling as $\binom{m}{k}$ in the collision-free regime, the learnable degrees of freedom are still constrained by the underlying linear-optical transformation, which has only $O(m^2)$ parameters. Thus, we can conclude that universality requires exponential resources in some form (qubits, gates, depth, or additional structure), which is consistent with results obtained for IQP-QCBMs~\cite{Kurkin2025IQP}. 

These observations motivate introducing an \emph{extended} model equipped with a fixed postprocessing map $f$ which compresses measurement outcomes to a lower-dimensional bitstring space, thereby inducing distributions over $\{0,1\}^n$, where $n$ matches the dimension of the target distribution. In \Cref{sec:universality}, we define this extended BSBM with fixed postprocessing and prove its universality. Before that, in the next section, we focus on efficient classical training of these models.

\section{Classical training of BSBMs}\label{sec:classical-training}

One widely used loss in generative modeling is the Maximum Mean Discrepancy (MMD), which compares two distributions by differences
of expectations in a reproducing kernel Hilbert space (RKHS).

\begin{definition}[Maximum Mean Discrepancy (MMD)]
Given two distributions $p$ and $q$ over a space $\mathcal{X}$ and a positive-definite kernel
$k:\mathcal{X}\times\mathcal{X}\to\mathbb{R}$ with associated RKHS $\mathcal{H}$, the MMD is
\begin{equation}\label{eq:MMD_def}
    \operatorname{MMD}(p,q)
    :=
    \sup_{\substack{f\in\mathcal{H}\\ \|f\|_{\mathcal{H}}\le 1}}
    \left(\mathbb{E}_{x\sim p}[f(x)]-\mathbb{E}_{y\sim q}[f(y)]\right).
\end{equation}
Equivalently, the squared MMD admits the kernel form
\begin{equation}\label{eq:MMD_sq_kernel}
    \operatorname{MMD}^2(p,q)
    =
    \mathbb{E}_{x,x'\sim p}[k(x,x')]
    +
    \mathbb{E}_{y,y'\sim q}[k(y,y')]
    -
    2\,\mathbb{E}_{x\sim p,\,y\sim q}[k(x,y)].
\end{equation}
\end{definition}

\begin{lemma}[Generalized kernel MMD via parity expectation values]\label{thm:generalized}
    For any stationary and bounded kernel $k$ over $\{0,1\}^m$, let $G$ be the Fourier transform of $k$. The $\mathrm{MMD}$ loss of BSBM as specified in \Cref{def:bsbm} can be expressed as:
    \begin{equation}\label{eq: MMD general kernel}
        \operatorname{MMD}^2(p, q_\theta) = \mathbb{E}_{\alpha \sim G} \left(\langle \Pi_\alpha \rangle_p - \langle \Pi_\alpha \rangle_{q_\theta}\right)^2,
    \end{equation}
    where $\Pi_\alpha \coloneqq \bigotimes_{j=1}^m \Pi_j^{\alpha_j}$ and $\Pi_j$ is the single-mode parity operator $\Pi_j \coloneqq (-1)^{\hat{n}_j}$ with $\hat{n}_j$ being the number operator on mode $j$. Moreover, a parity word expectation values can be written as
    $
        \langle \Pi_\alpha \rangle_p = \mathbb{E}_{x \sim p}[ (-1)^{\alpha \cdot x}]
    $ and $\langle \Pi_\alpha\rangle_{q_\theta} = \bra{\psi_{\mathrm{in}}} \hat{U}^\dagger(\theta) \Pi_\alpha \hat{U}(\theta) \ket{\psi_{\mathrm{in}}}$.
\end{lemma}
\noindent The proof goes analogously to the one presented in Ref.~\cite{rudolph2024trainability}, replacing Pauli-Z operators with parity operators.

Using Gurvits' algorithm~\cite{gurvits2005complexity,aaronson2011computational}, one can estimate the transition amplitudes of passive linear optical circuits with absolute error $\varepsilon$ with success probability $1-\delta$ using $\mathcal{O}(k^2 \log \delta^{-1} / \varepsilon^2)$ resources for collision-free Fock states consisting of $k$ photons. Consequently, phase-shifter expectation values of the form $\bra{\phi} \hat{U}^\dagger e^{i \sum_{j=1}^m \phi_j \hat{n}_j} \hat{U} \ket{\psi}$ can also be approximated for an arbitrary linear optical circuit $\hat{U}$. Noticing that the parity operator can be written as $\Pi_j = e^{i \pi \hat{n}_j}$, we can use this result to estimate the expectations of parity words as an ingredient for efficient MMD estimation on classical hardware.

\begin{lemma}[Estimating parity word expectations in a BSBM]\label{lemma: Efficient estimatmation of Pauli-Z word}
Consider a BSBM on $m$ modes as specified in \Cref{def:bsbm}. Let $q_\theta$ be the induced output distribution of BSBM. Moreover, let $\alpha \in \{0,1\}^m$ and consider the parity word $\Pi_{\alpha}$.
There exists a randomized classical algorithm that outputs an unbiased estimator $\widehat{\mu}_\theta(\alpha)$ of the expectation value $\mu_\theta(\alpha) := \langle \Pi_\alpha\rangle_{q_\theta} = \bra{\psi_{\mathrm{in}}} \hat{U}^\dagger(\theta) \Pi_\alpha \hat{U}(\theta) \ket{\psi_{\mathrm{in}}}$ such that, for any $\varepsilon,\delta \in (0,1)$, $\Pr\!\left[\big|\widehat{\mu}_\theta(\alpha)-\mu_\theta(\alpha)\big|\le \varepsilon\right]\;\ge\; 1-\delta$.
Moreover, the runtime of the algorithm is $\mathcal{O}(k^2\log\delta^{-1}/\varepsilon^2)$.
\end{lemma}

\noindent We provide a specification of the algorithm in Appendix~\ref{app:algorithm}. Furthermore, since in the collision-free regime the parity operator is equivalent to the pseudo-Pauli-$Z$ operator defined as $Z=\ket0\!\bra0-\ket1\!\bra1$ in the dilute regime, this algorithm can be directly applied to obtain the expectation values of words formed from these operators as well.

\begin{theorem}[Efficient estimation of MMD losses and gradients]
Let $q_\theta$ be the distribution generated by the BSBM and let $p$ be a target distribution. Consider the squared MMD with respect to a kernel $k$, and let $G$ denote the spectral measure associated with $k$ via the Fourier transform. 
There exists an efficient classical algorithm that constructs an unbiased estimator of $\mathrm{MMD}_k^2(p,q_\theta)$ and an unbiased estimator of its gradient with respect to $\theta$.
\end{theorem}
\begin{proof}[Sketch of proof]
     This is because of (1), the expectation value $\langle \Pi_{\alpha}\rangle_{q_\theta}$  can be efficiently estimated, and (2) $\mathrm{MMD}^2$ is a probabilistic mixture over these expectation values, which can be estimated efficiently by sampling bitstrings $\alpha\in\{0,1\}^m$ from the measure $G$.

    For gradients, note that the estimator for $\langle \Pi_\alpha\rangle_{q_\theta}$ as in Lemma~\ref{lemma: Efficient estimatmation of Pauli-Z word} can be written as a Monte Carlo average $\widehat{\mu}_\theta(\alpha)=\frac{1}{N}\sum_{i=1}^N X(\theta,\alpha;\xi_i)$, for some per-sample computation $X$ and internal randomness $\xi_i$. Under standard regularity conditions, differentiation can be exchanged with the expectation, giving $\nabla_\theta \mu_\theta(\alpha)=\mathbb{E}_\xi\!\left[\nabla_\theta X(\theta,\alpha;\xi)\right]$. In practice, $X$ can be implemented via differentiable linear algebra and automatic differentiation, whereas $\alpha$ and the internal randomness $\xi$ are sampled outside the automatic differentiation graph.
\end{proof}

\begin{remark}[Trainability does not imply simulability]
    Efficient estimation of MMD (and its gradient) does not imply efficient classical sampling from $q_\theta$. This is the same separation exploited in classically trained, quantum-deployed IQP models \cite{Kurkin2025IQP,recio2025train}.
\end{remark}

\begin{remark}[Unbiased estimator]
    Using \Cref{eq: MMD general kernel} directly along with the algorithm given by \Cref{lemma: Efficient estimatmation of Pauli-Z word} yields a biased estimator for $\operatorname{MMD}^2$. To make it unbiased, one should use the unbiased estimator from Refs.~\cite{gretton12, recio2025train}.
\end{remark}

\section{Extended mode: universality, expressivity, hardness and classical training} \label{sec:universality}

\subsection{Extended BSBM model}

As shown in \Cref{sec:preliminaries}, the BSBM with $n$ modes class is not universal for generative modeling over $n$-bit outputs, which is easily evidenced by a simple parameter counting argument. To increase the capacity to represent more distributions, we can enlarge the underlying photonic system, which increases the number of modes and available parameters.  
This will yield samples in a higher-dimensional outcome space, but which can be projected down to $n$-bit outcomes, while increasing the overall representation power. This is done by applying a fixed post-processing map (which we later call the readout map) with domain in the output space of the larger interferometer ($m$-dimensional bitstrings of weight $k$), and codomain in $n$-bit bitstrings. A related idea was used in~\cite{Kurkin2025IQP} to make the IQP--QCBM model universal: hidden qubits, which introduce extra degrees of freedom, were added, followed by a fixed readout (given by a partial trace over the hidden qubits). Following this intuition, we define an \emph{extended} BSBM generative model as follows.

\begin{definition}[Extended Boson Sampling Born Machine (EBSBM)]\label{def:EBSBM}

Am \emph{Extended} Boson Sampling Born Machine (BSBM) is a generative model specified by: (i) the dimension of the output distribution $n$ (here we focus on $n$-bit bitstrings), (ii) the number of optical modes used $m$, (iii) the linear optical network specifying the boson sampler specified by the unitary of the optical network $U \in \mathcal{U}(m)$,  (iv) a fixed initial state of $k$ photons (we will assume $m \gg k^2$ to remain in the diluted regime), and (v) a readout map $f:\Omega^k_m \rightarrow \{ 0,1\}^n$.
The initial state is fixed to be the $k$-photon input state $\ket{\psi_{\mathrm{in}}}=\ket{1^k0^{m-k}}$ that is, one photon in each of the first $k$ modes and vacuum in the remaining $m-k$ modes.
By a photon number measurement of all output modes, postselecting on the no-collision case, and applying the map $f$ on the output, it represents a family of distributions
$\{0, 1\}^n$:
\begin{equation}
    \mathcal{Q}^{\star}:=\{q^{\star}_{m,k,U, f}\}_{U \in \mathcal{U}(m)},
\end{equation}

In other words, if we denote the distributions obtained by the `bare' boson sampling model without the application of the readout map by  $q_{m,k,U}$ (defined by \Cref{eq: basic model induced probability} with outcome space $\Omega_m^k:=\{\bm s\in\{0,1\}^m:\|\bm s\|_1 = k\}$), then the extended model realizes the distributions on $\{0,1\}^n$ given by
\begin{equation}\label{eq:model distribution}
      q^{\star}_{m,k,U, f}(x) :=
      \sum_{\substack{\bm s\in\Omega_m^k\\ f(\bm s)=x}} q_{m,k,U}(\bm s).
\end{equation}
While $q^{\star}_{m,k,U, f}$ above is defined as a distribution, in context, we will also use it to refer to the underlying EBSBM realizing it.

\end{definition}

In the following, when we wish to focus on EBSBMs with exactly $m$ modes and $k$ photons, we will write $\mathcal{Q}^{\star}(m,k)$.

\begin{remark}[Fixed readout]
We note that universality could have been nearly trivially achieved by allowing parametrized readout maps, as they can convert even fixed sources of randomness to arbitrary distributions. However, in our case, the mapping is fixed.
\end{remark}

This definition above yields a simple universal construction we give more explicitly later but sketch here quickly: we fix $n\in\mathbb{N}$ and consider the EBSBM with a single photon ($k=1$), $m=2^n$ modes, and a readout map $f:\Omega^1_{2^n}  \rightarrow \{0,1 \}^n$ which returns the position of the measured single photon in binary. Then it is easy to see that the model is universal (see beginning of proof of  \Cref{thm:sufficient conditions}).

Although universality can be achieved at sufficiently large capacity parameter values, the required overhead is impractical as we need to specify exponentially large objects in $n$. In practice, we are interested in intermediate regimes that use fewer parameters but still cover rich classes of distributions, which motivates considering a family of models, characterized by a single\footnote{One can consider more parameters but for our purposes one will suffice.} structural capacity-increasing hyperparameter.
Analogous constructions are ubiquitous in machine learning, e.g., the one hidden layer neural network becomes universal as we increase the number of  hidden units which then constitute the capacity-increasing hyperparameter. 

\begin{definition}[EBSBM family with capacity hyperparameter]
    For each target output size $n$, we define EBSBM \emph{tower} as a sequence $\{\mathcal{Q}^{\star}(m_j(n), k_j(n))\}_{j\in J(n)}$, where each model is indexed by a capacity hyperparameter $j$ and $J(n)$ is the size of the family. Note that $(m_j(n),k_j(n))$ is chosen so that the corresponding model for each level $j$ is in the collision-free regime.
\end{definition}

We note that in general $J(n)$ may not be bounded, and the family could obtain universality only in a limit. In the examples we will construct we obtain exact universality for finite (but exponential) $J(n)$.

We aim to design a family satisfying the following \textbf{properties}:

\begin{enumerate}\label{it:setup}
  \item[P1.] \textbf{Monotonic increase in expressivity.} As $j$ increases, the model should become more expressive.

  \item[P2.] \textbf{Universality in the limit.} There should exist a capacity hyperparameter $j$ such that the model can represent any distribution on $\{0,1\}^n$. In general, it would also suffice to achieve universality only asymptotically.

  \item[P3.] \textbf{Potential for classical sampling hardness via reduction.} For polynomial-size model descriptions, the family should inherit worst-case sampling hardness from standard boson sampling.
  
\end{enumerate}

We will say that any family which has the 3 properties above has the \textbf{the monotonic, universal, hard (MUH) property}.
We note that the above properties pertain just to the properties of these families in terms of distributions they can represent. We will address learnability separately in a later section.

Next, we identify concrete criteria on the quantum model family which ensure these properties are achieved and provide concrete constructions.

\subsection{Achieving monotone expressivity, universality and hardness inheritance}

As each EBSBM family is characterized by a sequence of boson samplers and a family of matching readout maps $f$, next we find  sufficient conditions the the boson samplers, and of the readout maps which will guarantee that the three above conditions are satisfied.

\begin{theorem}[Sufficient conditions]\label{thm:sufficient conditions}
For each target output size $n$, consider a sequence of EBSBM models $\{\mathcal{Q}^{\star}(m_j(n), k_j(n))\}_{j=1}^{J(n)}$ with a sequence readout map $f_{j,n}:\Omega_{m_j(n)}^{k_j(n)}\to\{0,1\}^n$ (we will omit further parameter $n$). Assume that, for each level $j$, the model structure and the readout map satisfy the following conditions:
\begin{enumerate}
    \item $2^{n-1} < \binom{m_1}{k_1} \le 2^n$ for $j=1$.
    \item For each $j \ge 2$: $2^n \le \binom{m_j}{k_j} \le 2^{n+1}$ and the map $f_j : \Omega_{m_j}^{k_j} \to \{0,1\}^n$ is surjective.
    \item For each $j \in \{1,\dots,J-1\}$ there exist: (a) an efficiently computable embedding of unitaries $\iota_j : U(m_j) \to U(m_{j+1})$ and (b) an injective embedding of outcomes $R_j : \Omega_{m_j}^{k_j} \to \Omega_{m_{j+1}}^{k_{j+1}}$ such that:
    \begin{equation*}
        \forall U \in \mathcal{U}(m_j), \qquad 
        q_{m_{j+1},k_{j+1},\iota_j(U)} = (R_j)_{\#} q_{m_j,k_j, U},
    \end{equation*}
    and the readouts are compatible: $f_{j+1} \circ R_j = f_j$ on $\Omega_{m_j}^{k_j}$.

    Here $(R_j)_{\#} q_{m_j,k_j, U}$ denotes the pushforward distribution obtained by applying $R_j$ to the output sampled according to $q_{m_j,k_j, U}$.
    
    \item For level $j=1$ sampling from the raw distribution $q_{m_1,k_1, U}$ is intractable for classical polynomial-time algorithms under standard boson-sampling complexity theory assumptions. 
    
    \item Define injections $F'_j : \Omega_{m_1}^{k_1} \to \Omega_{m_j}^{k_j}$: by composing the embeddings:
    \[
    F'_1 := \mathrm{id},
    \qquad
    F'_j := R_{j-1} \circ R_{j-2} \circ \cdots \circ R_1
    \quad (j \ge 2).
    \]
    For every $j \ge 1$:
    \begin{enumerate}
        \item the restriction $f_j\big|_{F'_j(\Omega_{m_1}^{k_1})}: F'_j(\Omega_{m_1}^{k_1}) \to f_j(F'_j(\Omega_{m_1}^{k_1}))$
        is bijective, and the inverse \\
        $( f_j\big|_{F'_j(\Omega_{m_1}^{k_1})})^{-1}$ is computable in polynomial time.
        \item membership in $f_j\left(F'_j(\Omega_{m_1}^{k_1})\right)$ is decidable in polynomial time.
    \end{enumerate}
    \item There exists a level $J$ such that: (i) the 
    underlying boson sampling interferometer can embed an interferometer with $2^n$ modes and exactly 1 photon in chosen $2^n$ modes, (and which moves the remainder of the photons to other chosen $k_J-1$ modes) and (ii) $f_J$ maps the possible $2^n$ distinct output states of this setting bijectively to $\{0,1\}^n $ and where both $f_J$ and its inverse are computationally efficient.

\end{enumerate}
Then, this family satisfies \textbf{properties P1-P3}, i.e., has monotonically increased expressivity, universality-in-the limit, and inherited hardness from the base level, i.e., \textbf{has the MUH property}.
\end{theorem}

Before providing the proof, we first remark on the reasons for such a complicated construction. More precisely we briefly highlight the technical difficulties in achieving the MUH property.
Concerning hardness, the anchor we use is the assumption of the hardness of sampling from the bare boson sampler. However, the readout map changes the situation a bit. Specifically, non-injective readout maps can provably reduce the hardness (e.g., a trivial readout map which projects all outputs to a single output gives a trivial-to-sample-from distribution).

Conversely, injectivity, together with additional computational assumptions, i.e., efficiency of $f$ and its (partial) inverse would ensure hardness: in this case we can apply $f^{-1}$ to the output space of the full model to recover the distribution from the underlying boson sampler. 

The problem is, universality demands surjectivity. However, the output set of boson samplers and the set of bitstrings cannot be equipotent, except in extreme cases where $k\in \{0,1,m-1,m\}$, which are not suitable for hardness arguments.
The solution we utilize is to embed a ``smaller'', but still hard, boson sampler (with output space cardinality below $2^n$) into a larger one, and insist that the readout map be injective when restricted on outputs produced by this embedding; efficient inverse there together with the capacity to decide whether the output stems from an embedded smaller interferometer suffices then for hardness. These arguments motivate conditions 1 and 5.
At the same time we can use the embedding property to prove monotonicity of expressivity, provided the capacity of the readout map to maintain compatibility as the sizes of the underlying interferometers grow.

Now we can provide a sketch of the proof of the theorem.

\begin{proof}[Sketch of proof]

Regarding Property P2, universality at large hyperparameter settings, this will follow immediately from 6.
The idea is as follows. If we consider a linear optical network with $2^n$ modes specified by an arbitrary $U\in \mathcal{U}(2^n)$, given a Fock state with one photon in the first mode on input we see that the output is $\ket{\psi_{\mathrm{out}}} = \sum_{k=0}^{2^n-1} U_{k1} a^\dagger_k \ket{0}$ which is an arbitrary quantum state, and thus measurements realize an arbitrary distribution over Hamming-weight-1 bitstring. Choosing the readout map to decode this to binary (by giving the binary specification of the position where the photon was measured) realizes a universal model. We note that this is a special case of a map which outputs the lexicographic rank which we will encounter later again.

 Regarding Property P1, condition 3. above directly implies that $j^{th}$ element of the sequence of the models can express all distributions as the $(j-1)^{st}$, proving inclusion of distributions as we increase the hyperparameter $j$. This does not imply a strict inclusion, but this can be attained as we show in a subsequent example.  Regarding Property P3, the conditions 3.-5. ensure the possibility of a reduction of sampling from the first (assumed to be hard) model to the sampling from the $j^{th}$. Specifically: (i) from condition 3. we have that the distributions parametrized by the $j=1$ model are classically intractable; 
(ii) from condition (4) we have that the model at every $j^{th}$ level can express distributions at every previous level including the first via a sequence of $j$ efficiently computable embeddings; and (iii) from condition (5), given a sampler that can simulate sampling from all parameter settings of the model at the $j^{th}$ level, we can efficiently simulate the sampling at the first level (as $f^{-1}$ has an efficient inverse at the outputs stemming from the sequence of embeddings).

We note that the efficiency of the cumulative embedding (which realizes a mapping from the description of the first model to the $j^{th}$, and which is required to be efficient for the hardness argument to hold) is lost if $j$ is superpolynomial. However in this case we already are considering models with superpolynomially many parameters, which are inherently not efficient to compute, classically or quantumly.\end{proof}

We note that \Cref{thm:sufficient conditions} provides sufficient conditions on the family of boson samplers and readout maps which are rather abstract, and at the face of it, it is not clear whether they can be satisfied. 
Next we prove that they can be satisfied, even with a strict increase in expressivity (strict inclusion of distribution families) via a more explicit construction

\begin{lemma}[Existence of MUH family via gradual bleed-mode construction]\label{LM3}
     For each target output size $n$, there exists a sequence of EBSBM models $\{\mathcal{Q}^{\star}(m_j(n), k_j(n))\}_{j=1}^{J(n)}$ specified with a sequence of boson samplers and a sequence of  readout map $f_{j,n}:\Omega_{m_j(n)}^{k_j(n)}\to\{0,1\}^n$ which satisfy the 6 conditions of  \Cref{thm:sufficient conditions}
.
\end{lemma}
\begin{proof}
We omit the index $n$ for clarity unless explicitly needed.
We will define a sequence of boson samplers specified by number of modes $m_j$ and photons in the initial state $k_j$ as follows.    
\begin{enumerate}
    \item Base level. The pair $(m_1,k_1)$ chosen such that it is in the (diluted) collision-free regime and satisfies $2^{n-1}<\binom{m_1}{k_1}\le 2^n$. Moreover, we choose $(m_1, k_1)$ be such that the sequences $(m_1(n), k_1(n))$ relative to the size of the ouptut distribution specify a family of hard-to-simulate boson samplers, while $m_1(n)$ is still just polynomial in $n$. We comment on the existence of such a sequence later. 

    \item Threshold crossing at level $2$. One sets $k_2:=k_1$ and chooses $m_2>m_1$ minimal such that $\binom{m_2}{k_2}\ge 2^n$.

    \item Gradual growth for $j\ge 3$. For each $j\ge 3$, 
    
    the pair $(m_j,k_j)$ remains in the dilute regime, satisfies $\binom{m_j}{k_j}\ge 2^n$,
    and admits a decomposition $m_j=M_j+b_j$, where $M_j$ is the number of \emph{working modes} and $b_j$ is the number of \emph{bleed modes} satisfying $b_j \geq k_j-1$. .
    
    The sequence $(M_j)_{j=3}^J$ is non-decreasing, with $M_j<2^n \quad (j<J) $ 
    and $
    M_J=2^n.$ Finally we set  $k_{j+1} = k_{j}+1$ for $j>1$
    \end{enumerate}
Before specifying the readout map family, let us explain the construction intuitively. 
Points 1. and 2. from the specification above directly match the condition 1. and 2. from Theorem \Cref{thm:sufficient conditions} and have been motivated earlier. 
The bleed modes will enable an explicit and straightforward embedding of a smaller boson sampler into the larger one.  We use the bleed modes  to reduce the number of photons in ``working modes'', which can then exactly reproduce the smaller beamsplitter.
At the same time, the addition of a new photon into the whole system and which we can choose not to move to the bleed modes in general, allows for a larger space of distributions, if the postprocessing/readout map is chosen correctly. 
Further at each step we also increase the number of non-bleed (working) modes.
While the addition of new working modes is not neccessary to have monotonicity (bleed modes and more photons suffice), the addition of the models allows us to give a simple proof of exact universality when $j$ becomes exponential.

More precisely 1-3. allow for the following embedding maps. The boson sampler at level $1$ embeds into level $2$ by mode padding:
    \begin{equation*}
        \iota_1(U)=U\oplus I_{m_2-m_1},
    \end{equation*} 
    with the outcomes embedding as
        $R_1(x)=x\|0^{\,m_2-m_1}$, where $\|$ denotes concatenation of strings \footnote{In other words, we simply embed the smaller interferometer into the larger, and simply set it up such that it forwards all the new modes to the last modes of the output. These will always be vacuum modes when embedding the model from the first step.  }. 
Other boson sampler embeddings $\iota_j(U)$ are defined as follows, operatively: given $k_{j+1}$ photons at level $j+1$, the interferometer is set up such that the $k_{j+1}^{th}$ photon is rerouted to the new (and last) bleed mode. What remains is a system with  $k_{j}$ photons, and $m_{j+1} -1 \geq m_{j} $ modes available for optical processing. Now we simply realize the smaller boson sampler in the available modes modes via direct copying of linear-optical operations. Via this construction, if the smaller boson sampler outputs the bitstring $x$, the larger outputs $x \| \underbrace{0 \ldots 0}_{l_j} \| 1$, where $l_j = m_{j+1} - m_{j-1}-1$. This also defines the readout embedding maps: $$R_j(x) = x \| \underbrace{0 \ldots 0}_{l_j} \| 1$$

Next, we will provide a specification of the readout maps. Our specification will be again somewhat abstract and will only specify the sufficient conditions (and when not obvious, explain why they can be met) on the maps which prove the Lemma, leaving the possible realization choices for a later discussion. 

We begin with $f_1: \Omega_{m_1}^{k_1} \rightarrow \{0,1\}^n$ which is just any efficiently computable injective map, with an efficient inverse, and where the membership in $f(\Omega_{m_1}^{k_1})$ can be efficiently decided. A simple example proving realizability is the map which outputs the lexicographic rank of the $m_1$-bit Hamming weight $k_1$ bitstring given on input in binary. Since $\binom{m_1}{k_1} \leq 2^n$, this map exists. 

For $j>1$, the maps $f_j$ need to satisfy conditions 3 and 5. 
Condition 3 is a readout compatibility constraint, $f_{j+1} \circ R_j = f_j$. 
Whenever $f_j$ is a feasible map, this is always possible to achieve by using the above relation as the definition of the action of $f_{j+1}$ on the subset of inputs $R_j(\Omega_{m_j}^{k_j}) \subseteq \Omega_{m_{j+1}}^{k_{j+1}}$.
Next, Condition 5 requires us to consider the specification of the maps $F'_j$ which are a composition of all output embedding maps $R_j$. In our case it is simple: for an $x \in \Omega_{m_1}^{k_1}$ we have $F'_j(x )=  x \| 0\ldots 0 \| \underbrace{1\ldots 1}_{j-1}$. 
So now we need to define $f_j$ to be bijective when restricted on bitstrings of the form $ \{ x \| 0\ldots 0 \| \underbrace{1\ldots 1}_{j-1} \ \ \ \vert x \in \Omega_{m_1}^{k_1}\} \subseteq \Omega_{m_j}^{k_j}$. A simple example is again choosing this map to output the lexicographical rank of $x$ in binary, which will also be consistent with the compatibility Constraint 3 when chained.  
In this case it is easy to see that (on this set of inputs), the map $f_j$ is efficient, with an efficient inverse and an effcient membership check. 

Condition 4 follows explicitly from the point 1 in the Lemma.

What remains to check is Condition 6 but this follows straightforwardly. At $J=2^n$ we have more than $2^n$ ``work'' modes (as we added at least one at each step, and started at $m_1$), $k_J$ photons, and at least $k_J -1$ bleed modes. Thus we can realize a single photon $2^n$ mode interferometer in the first $2^n$ modes, while rerouting all the other photons to the bottom modes. It suffices to fix $f_J$ such that on inputs which have only one photon in the top $2^n$ modes it outputs the lexicographic rank in binary. 

Finally, we remark on the consistency between  condition 4 (toward hardness) and the choice of $m_1,k_1$ such that $2^{n-1} \leq \binom{m_1}{k_1} \leq 2^{n}$.

While the first harness analyses of sampling from boson samplers assume particular relationship between $m$ and $k$ as parts of proof techniques, e.g., $m \approx n^5 + \text{polylog}(n)$, we can safely assume hardness in a much more wider set of relationships between the photon number and mode number, in the particular case when we talk about \textbf{worst-case hardness}, which is what we focus on here.

Suppose (worst-case) hardness holds for boson samplers with $k$ photons over $m$ modes in a regime defined by some relation $k=h(m)$. Assuming the non-collapse of the polynomial hierarchy, and additional standard assumptions for the hardness arguments needed in the noisy setting \cite{aaronson2011computational}, this implies that there is no classical sampling algorithm with a polynomial run-time in $m$, when $k$ is set as $h(m)$ in the worst case, i.e., for all such instances. 
Consider now a perturbed setting where there exists a polynomial-time algorithm which can achieve a polynomial runtime in $m$ robustly in the worst case\footnote{Average-case statements could also be achieved straightforwardly for non-natural distributions of inputs. We make no claims of average-case hardness here however.}, when the photon number is chosen as some other function $h'(m)$.
We then note that we can embed an $(m,h'(m))$-boson sampler into a $(m', h(m'))$-boson sampler if we choose $m'$ such that $h'(m) = h(m')$. Defining $m'=g(m)$, the embedding condition is now $h'(g(m)) = h(m)$. If $h$ and $h'$ are polynomially related, then $g$ is a polynomial (or an inverse of a polynomial). Then, if there exists a polynomial time algorithm for sampling from a $(m,h'(m))$ boson sampler, there exists a polynomial time algorithm for sampling from a $(g(m), h(g(m)))$ boson sampler. Our assumption was that no polynomial time classical sampler for a $(m,h(m))$ boson sampler exists.

This implies that, in the worst case, even for an inverse-polynomial $g$, sampling from any $(g(m), h(g(m)))$ instance cannot be done classically in polynomial‑time in $m$ without contradicting the hardness of the $h(m)$-regime.

In other words, we can assume hardness for a much broader range of relationships between $m$ and $k$, and in this range it is always possible to find  $m_1,k_1$ such that $2^{n-1} \leq \binom{m_1}{k_1} \leq 2^{n}$.

 \end{proof}

The construction in  \Cref{LM3} gives MUH properties, but the monotonicity is not \emph{proven} to be strict -- although since in the limit we reach full universality, the inclusion must be strict at at least one level.   In the next Corollary we give additional conditions on the readout maps which when met guarantee strict inclusion at each step.

\begin{cor}[MUH with strict increase in expressivity]\label{COR3}
    The construction from  \Cref{LM3}  can also achieve strict inclusion, i.e. a strict increase in expressivity at each step. 
\end{cor}

\begin{proof}
Here, the result will be achieved not just by manipulation of the readout maps, but by using properties of the linear optical network.

Consider the $j^{th}$ and the $(j+1)^{th}$ level in our construction. 
By the embedding constructions all outputs of the $j^{th}$ boson sampler are mapped injectively to outputs of the $(j+1)^{st}$ via map $R_j$, and we know that all distributions that the boson sampling model at level $j$ can be realized in the $j+1$ model up to the renaming/relabeling given by $R_j$. We also know that $R_j(\Omega_{m_j}^{k_j})$ by construction has a photon deterministically in the last bleed mode.
The linear optical network at $(j+1)^{st}$ level which embeds the $j^{th}$ beamsplitter, however, can be augmented, at the very last layer of gates, to re-inject this photon from the last bleed mode to some of the other modes with a desired probability by, e.g., controlling the transmitivity of a beamsplitter \footnote{In principle this could cause a collision, but since we are in a low collision regime overall, a collision will in general be unlikely.}. Call this probability of injecting the photon back $p$.  

This new parameter $p$ causes a mixing of distributions: with probability $1-p$ (not re-injected) we realize the distribution from the set reachable by the boson sampler at $j^{th}$ level. With probability $p$ however, we are guaranteed to return an output not in the set  $R_j(\Omega_{m_j}^{k_j})$ .
The consistency conditions $3$ fix what the map $f_{j+1}$ does to an output  resulting from the case of non-injection. However, for all the inputs to $f_{j+1}$ outside $R_j(\Omega_{m_j}^{k_j})$, we are free to define the output, and this will result in the mixing of two classes of distributions, which intuitively can be larger than either.
We prove this is the case. 

In particular, we will define $f_{j+1}$ to given any output stemming from the mixing event (e.g., whenever the boson sampling measurement is not in $R_j(\Omega_{m_j}^{k_j})$) to output a fixed bitstring $b^\ast\in \{0,1\}^n$ which we fix shortly, albeit non-constructively.

Let $\mathcal{D}_j$ be the set of distributions realizable by the model at the $j^{th}$ level, and assume that it does not include all distributions. Then there exists a bitstring $b\in \{0,1\}^n$, a probability distribution $q \in \mathcal{D}_j$ and a probability $p\in [0,1]$ such that the mixed distribution $p \times \Delta_b + (1-p) \times q \not \in \mathcal{D}_j$. This we prove more generally in the Lemma below.
For our construction, we then choose $b^\ast = b$ as in the claim above. This construction is not in contradiction with any of the 6 properties we achieved earlier as only fixes what the readout maps $f_j$ do outside any of the circumstances considered in \Cref{thm:sufficient conditions}.
This proves that at each level we at least obtain one new distribution which was not accessible before, thereby proving a strict inclusion of the set of distributions we can capture. 
\end{proof}

We now prove the technical lemma used above.

\begin{lemma}
Let $\mathcal{D}$ any set of distributions over $\{0,1\}^n$, and assume that it does not include all distributions. Then there exists a bitstring $b\in \{0,1\}^n$, a probability distribution $q \in \mathcal{D}$ and a probability $p\in [0,1]$ such that the mixed distribution $p \times \Delta_b + (1-p) \times q \not \in \mathcal{D}$, where $\Delta_b$ denotes the Kronecker delta distribution with unit mass on $b$.
\end{lemma}
\begin{proof}
By contradiction. Assume the opposite. This implies that for all $b$,  $\Delta_b \in \mathcal{D}$, by picking any $b$ and taking $p=1$. 

Given a $q \in \mathcal{D}$, a  $b\in \{0,1\}^n$ and a probability $p$, let us call the mapping $q \rightarrow p \times \Delta_b + (1-p) \times q$ \textit{the shifting of $q$ in the direction of $b$ by p}.
Assuming the opposite, it follows that $D$ is closed under any set of shifts in any bitstring direction by any $p$.
However it is easy to see that starting from a suitable Kronecker delta (on any bitstring), we can realize any distribution $q^\ast$ by a sequence of $n-1$ shifts. To see this, we think of distributions as probability vectors. Fix an arbitrary target $q^\ast$ . We start in a Kronecker delta in a mode of $q^\ast$, and then we sequentially `add' weight to every element of the support, from most probable to least probable breaking ties arbitrarily. Specifically, at each step we add sufficient mass to a new element to ensure that the relative mass between the already present elements and the new element matches the same ratio as in the target distribution. This converges to the target in $n-1$ steps.
But then $\mathcal{D}$ includes all distributions, which is a contradiction. This means there must exist at least a single distribution for which a shift in at least one direction by some amount is not an element of $\mathcal{D}$.
\end{proof}

We note that the proof above is not constructive, as we do not provide an efficient algorithm to find such a $b$ at each step. At the same time, we conjecture that for almost all generic specifications of what $f_j$ does at each stage, we will have an increase of expressivity. 

The constructions above are all based on \Cref{thm:sufficient conditions}, which provides sufficient but rather convoluted conditions to achieve the desired MUH properties. They are interesting from a purely theoretical standpoint, but not necessary. Next we give a much simpler construction which is (arguably) more elegant, for which we do not have a (strict) inclusion property to achieve Property 1 (increase of expressivity), but it is likely on other grounds.
Since we do not have a proof of strict inclusion, hardness is also proven in a different way than simply leveraging \Cref{thm:sufficient conditions}.

This construction also starts off at a base level which is an assumed-to-be hard beamsplitter with suitable mode number and photon number $(m_1, k_1)$. The end of the tower is the single-photon construction with exactly $m_J=2^n$ modes and $k_J = 1$ photon. All the intermediary levels are obtained by a straightforward interpolation subject to constraint on size of the output space. 

\subsection{Construction by direct parameter interpolation}
We first start by specifying the sequence of boson samplers: for $j=1$ we define $(m_1, k_1)$ to be such that classical sampling is likely intractable from a model with $(m_1-1)$ modes and $k_1-1$ photons. 

For all $j>1$ we then set $k_j = k_{j-1} -1$, which ensures that at the end of the tower we have just one photon.
$m_j$ are chosen relative to $k_j$ such that for all $j>1,$
$\binom{m_j-1}{k_j-1} \le 2^n \le \binom{m_j}{k_j} \leq 2^{n+1}$.

We note that if $2^n \le \binom{m_1}{k_1} \leq 2^{n+1}$ holds then there exists a sequence of $m_j$ which satisfy this condition \footnote{The increase of the binomial coefficient by moving from $m$ to $m+1$ is very slight so this follows.}.
Combined with the remarks about the robustness of hardness arguments at the end of the proof of Corollary \ref{COR3}, it is safe to assume such a choice exists which is consistent with hardness assumptions. 
In other words, at each step we have the output space such that it allows a surjection onto $\{0,1\}^n$, whereas if we fix one photon to occupy the last mode, we can have an injection to $\{0,1\}^n$.
Now we can explicitly define the family of readout maps.

We are assuming $\binom{m-1}{k-1} \le 2^n \le \binom{m}{k} \leq 2^{n+1}$. 
Let $\Omega_{m}^{k}:=\{s\in\{0,1\}^m:\ |s|=k\}$.

Define the readout \(f_j:\Omega_{m_j}^{k_j}\to\{0,1\}^{n} \) as follows (we will omit the subscripts $j$ to deload notation):

\begin{equation}
    f(s)=
    \begin{cases}
        \mathrm{bin}_n \bigl(\mathrm{rank}_{m-1,k-1}(s_{1:(m-1)})\bigr),
        & s\in E_1,\\
       G( \mathrm{bin}_{n+1} \left(\binom{m-1}{k-1}+\mathrm{rank}_{E_1}(s))\right),
        & s\in \Omega_{m}^{k} \setminus E_1,\\
     
    \end{cases}
\end{equation}
where $E_1:=\{s\in\Omega_{m}^{k}: s_m=1\}$.

Here, $\mathrm{rank}_{E_1}: \Omega_{m}^{k} \setminus E_1\rightarrow \{0, \ldots , \binom{m}{k} - \binom{m-1}{k-1} -1 \}$ is the lexicographic rank of the input bitstring but skipping all the elements in $E_1$. 
In general  $\binom{m-1}{k-1}+\mathrm{rank}_{E_1}(s)$ can be larger than $2^n$ and we wish $n-$bit outputs, so the map $G$ compresses the output to $n$ bits, e.g. by ignoring the most significant bit.

It is instructive to explain what $f$ does in words. To ensure hardness, we will embed a $(m-1,k-1)$ interferometer into a larger boson sampler, and reroute one photon to the last ($m^{th}$) mode.
Now all bitstring outputs that end with a 1 in the last position correspond to possible valid outputs of the embedded interferometer - this corresponds to the set $E_1$ -  and we ensure $f$ is injective and bijective when restricted on them. This will allow us to conclude the inheritance of hardness of sampling from this model\footnote{We note that `inheritance' does not mean it is inderited from one level of the family of models to another, but rather that it is inherited from the hardness of sampling from boson samplers.}. 
We also note that at the top of the tower, we have $m=2^n$ and $k=1$. In this case, $E_1$ is a singleton (bitstring with the 1 in the last mode) which returns the all-zero bitstring.
For all other bitstrings it will return exactly the binary encoding of the position of the 1, ensuring universality.

We capture these properties with the following lemma:

\begin{lemma}
The construction by direct parameter interpolation inherits hardness from boson samplers for those sections of the construction which allow a polynomially-sized description. 
\end{lemma}

\begin{proof}
We note that in the construction the photon number $k,$ starting from superlinear in $n$ decreases linearly with $j$, whereas the size of the output space remains above $2^n$. This means $m_j$ grows rather rapidly, and in particular, in the regions of the tower where $k$ is sublinear in $n$ (e.g.\ constant), the mode number must be superpolynomial. In these regions, the interferometer is too big to even be efficiently specified (in the general case), and we will not analyze hardness of simulation there.

The practically relevant regimes are those where the mode number is polynomial. In these cases, the photon number must be at least linear in $n$, and by arguments of robustness in the proof of Corollary \ref{COR3}, there the hardness follows. But this also holds for shifts: if the family (now scaling with $n$) of boson samplers $(m(n),k(n))$ is classically intractable to simulate, then so is $(m(n)-1,k(n)-1)$ (by similar arguments as in the proof of Corollary \ref{COR3}).

Then for each $j$ in the relevant region, we note that by restricting the $(m_j,k_j)$-boson sampler settings which output a photon deterministically in the last mode, we realize an injective embeddding of a $(m_{j-1},k_{j-1})$ boson sampler with the output embedding $x  \stackrel{R}{\rightarrow} x || 0\ldots 0 || 1 $, where $x$ is Hamming weight $k_{j-1} = k_{j} -1$ similar to previous constructions. Then we notice that the mapping $f_j$ when restricted to bitstrings of the form $x || 0\ldots 0 || 1$ is an injective mapping onto $\{0,1 \}^n$, and moreover it is easy to check if a given $n$-bit bitstring is from the codomain of this restriction of $f_j$, and the inverse is efficient as well.
This means that the capacity to sample from the model at the $j^{th}$ level efficiently implies the capacity to sample efficiently from all $(m_{j-1},k_{j-1})$ boson samplers, and that is in contradiction with the assumption. 
\end{proof}

The hardness proof for this construction is more local: at each level, we show that a hard smaller boson sampler embeds into a restricted subset of the larger model, and require the postprocessing to be injective/invertible on that restricted subset. That gives hardness without needing the full tower-wide inheritance theorem.

The only remaining question is whether expressivity increases monotonically with the tower level $j$. For this construction, it is easy to see that a strict inclusion is not realized due to differences of cardinalities of output spaces, and due to a lack of an explicit embedding on the level of the boson samplers. 

In absence of strict inclusion, the idea of increasing expressivity becomes more challenging to define as one must commit to an expressivity monotone (a mapping from distribution family to $\mathbb{R}$ which is non-decreasing under inclusion). We do not venture in this direction in this work, but rather provide the following intuition: For each choice of structural hyperparameters, the model defines a smooth map from the unitary parameter space to the probability simplex. As the number of free parameters grows along the sequence, the effective dimension of the set of realizable distributions should also increases, which provides an intuitive sense in which expressivity should improve. A rigorous geometric argument, however, as well as a precise definition of an expressivity measure, is left for future work.

\subsection{Efficient classical training}
\label{sec:training}
In the previous sections we just focused on the aspects of the models as generators of distributions. 
In this section we consider how the shift from the bare boson sampling model to the extended model changes how the model can be trained. 

For clarity when talking about efficient training, we only refer to the \textit{capacity to classically estimate the loss and its gradients}; we make no statements on the trainability of the model in terms of number of e.g. vanishing gradients, or otherwise rate of convergence too good loss minima.

For the bare sampling model the key idea was that an MMD-type loss can be computed from the quantum model by estimating expectations of parities out outcomes.
In the extended model, we are dealing however with the push-forward distribution obtained by applying the readout map $f$ onto the output of the boson sampler. One approach to raining via MMD-type loss would require the capacity to evaluate parity expectations of the whole model efficiently. However, since the readout map is in general not linear, it does not commute with taking expectation values, and so the capacity to compute expectations of the bare boson sampler does not translate to capacity to efficiently estimate the expectations of the push-forward distribution. 

Here we take another approach:   we introduce a lifting procedure that maps an $n$-bit dataset into $\Omega_{m_j}^{k_j}$ so that training can be performed in the raw space.

\subsection{Pushforward, non-identifiability, and lifting}

Let $\Omega_{m}^{k}$ be output space of the bare boson sampler the and let $f:\Omega_m^k\to\{0,1\}^n$ be a readout map.
With $P_U(s)$ we denote the distribution over $\Omega_{m}^k$ of the outputs of the boson sampler ($U$ collects all the free parameters), and with 
\[
Q_U := f_\# P_U \quad \text{on }\{0,1\}^n,
\qquad Q_U(y)=\sum_{x:\,f(x)=y} P_U(x).
\]
we denote the final (pushforward) distribution. 

Instead of training the model in the final output space, we consider training in the space  of the bare boson sampling outputs.

To train in $\Omega_m^k$, we choose a (deterministic or stochastic, as we explain shortly) a right-inverse of $f$ denoted $f^{-1}$ which is such that $f(f^{-1} (x)) = x$, for all $x$ in the codomain of $f$.

The idea is then to map (or ``lift'') the training data via the map $f^{-1}$ and then train the bare boson sampling model on the lifted data. 

Such an $f^{-1}$ is not unique. More generally, we can define the global inverse $F: \{0,1 \}^n \rightarrow \mathcal{P}(\Omega_m^k)$, where $\mathcal{P}$ denotes the power set, and $F(x)$ is the set of all $y$ such that $f(y)=x$.
Each $f^{-1}$ is realized by picking one representative from each set $F(x)$.

\paragraph{Deterministic lift.}
For a deterministic lift to be possible, $f$ must allow at least one efficient (right-)inverse $f^{-1}$.
In the constructions we provided in previous sections, we already required $f$ to have an efficient inverse on parts of the output space to achieve hardness inheritence. For training we need the inverse to be efficient everywhere. In all the constructions we provide this is possible.
We note that we can use a different inverse for training than is used for arguments of hardness. 

\paragraph{Interpreting training via a deterministic lift.}
Recall that the standard classical training of the bare boson sampling model corresponds to training with respect to some $MMD_\kappa$ metric specified by a kernel $(x,y) \mapsto \kappa(x,y)$, where, technically $x,y \in \{ 0,1\}^m$, even though the model only outputs bitstrings in $\Omega_{m}^k$.  
It is tempting to attempt to express the deterministic lift as a different kernel obtained by composing the bare kernel with the inverse map $f^{-1}$: $k_{f^{-1}}(x,y) \mathrel{:}= \kappa(f^{-1}(x),f^{-1}(y) ) , $ which is now defined on the space $\{0,1\}^n$.
This would allow for a simple grasp of what the training procedure optimizes. However, the procedure we described does not correspond to training with respect to this  $MMD_{\kappa_{f^{-1}}}$ because we are training the bare boson sampler, and not its lifted variant: $P_U(s)$ instead of $f^{-1}\# Q_U(s) =f^{-1}  \# f \# P_U(s)$. We note that $f^{-1}$ is just a right inverse and not the left inverse, so $f^{-1} \circ f$ is not the identity. 
In theory this model is easily realized: we sample from the beamsplitter, apply $f$ and then the chosen $f^{-1}.$ However recall we want classical training so sampling is not an option. 

Consequently the process we described realizes a surrogate, an approximation of the $MMD_{\kappa_{f^{-1}}}$, which is only fully faithful when $f^{-1}$ is a both-sided inverse.
This does not mean that either option is better in practice regarding performance, however, just that we will not be able to easily theoretically analyze the performance of the model.

This situation allows for interesting features: for example it can happen that the training loss is not zero, even though the realized distribution is equal to the empirical distribution of the training set. However this loss is faithful in the sense that if the empirical loss is exactly zero in the limit of large training set sizes, then the distributions will match.

\paragraph{Stochastic lift.}

More generally, likely better training settings could be realized by applying a a stochastic right-inverse: we map each data point $x$ to a sample from the set of possible inverse values $F(x)$. This may be more difficult in practice, but since the bare boson sampling distributions are rather restricted for reasonable model sizes, this approach could allow smaller losses and better convergence.

\paragraph{Training procedure versus hardness.}

In the previous section we made no claims with respect to quantum advantage in generative modelling. Our only claim was that the models we define likely cannot be classically simulated for all parameter settings. 

A valid question arises with respect to the interplay between training and hardness. Specifically, is the training procedure such that it will encounter hard parameter settings in any scenario: converse would imply that even though the models could capture hard distributions, we will never encounter them in practice. 
While we will not prove the statement either way for our model we will briefly discuss possible points of concern.

In previous works studying similar questions \cite{coyle20} the arguments mostly concerned the likelihood random initializations or training steps encounter likely hard parameter settings. Here we consider the question: assuming the model is being trained on samples from a hard distribution (e.g. generated by one of the EBSBM models themselves), will it learn the correct (intuitively hard) parameters. More precisely, we ask: is the situation different than in the case of training raw boson sampling models.  
The problem is that in all our construction the boson sampler undergoes the readout map $f$, however, in all of them we made sure it is injective and invertible on some hard subset of distributions. 

Now, if we choose that same $f^{-1}$ for the lifting procedure \footnote{The lifting map  $f^{-1}$ which we obtain from the hardness constructions is not defined on all bitstrings, as it was only needed to be bijective between the output space of a smaller beamsplitter $\Omega_m^k$, which is in general of smaller cardinality than $2^n$. For our arguments it does not matter how the inverse is generalized to all bitstrings.}, then there is no difference in training the EBSBM model and the raw boson sampler. We note that the resulting model will nonetheless be in principle more expressive than the smaller beamsplitter even when restricting support just to the valid embedded outputs, so this choice is not in general problematic from an expressivity perspective either.
We will analyze other options for training in follow-up works. 

\section{Discussion and open problems} \label{sec:discussion}

In this paper, we have introduced families of boson-sampling-based generative model families, with a hyperparameter controlling the expressivity, which are universal in the large-mode limit and which have the potential for a sampling advantage. Further, we have introduced methods for the classical efficient training of these models as was previously achieved for IQP-based quantum circuit Born machines (IQP-QCBMs).

Our work opens a number of questions. Firstly, the extent to which quantum advantage could be expected from using these models is a matter of some debate. 
Our BSBM models are trained to learn a target Boson Sampling distribution via an MMD loss function, based on parity expectation values. On the other hand, the complexity-theoretic hardness results for approximate Boson Sampling are rather based on total variation distance \cite{aaronson2011computational}. At present it is unclear what can be said about hardness of approximate Boson Sampling, when the approximation is quantified by MMD with respect to specific observables. Specifically, it is unclear whether such models can be ``spoofed'' in the sense of there being classical models which can reach the same MMD distance as the quantum model to any target.
Another key question pertains to the trainability of these modes, but now in the sense of actual magnitude of the gradients, and the potential of the training to find good local minima. Recent work on gradients of IQP-QCBMs \cite{shen26} showed explicit gradients can be bounded in certain cases, and it is an open question if similar can be done for the models presented in this paper.

Furthermore, we have only studied generative models built around the very basic type of boson samplers and, for this reason, our expressivity-increasing and universality constructions are somewhat convoluted.
Simpler universality constructions and interesting classes of more expressive models could arise by utilizing boson samplers with additional enhancements such as with efficient postselection or adaptive measurements. While post-selection on a constant number of modes will not significantly affect our results, a natural open question is to investigate to what extend such minimal  augmentations  preserve or simplify trainability, expressivity and universality, while ideally lowering the size requirements for computational hardness.
Another degree of freedom which we fixed and could be generalized involves the choice of the initial state, where generalized Gurvits-like algorithms are available for estimating expectation values efficiently \cite{LimOh2025}.
Also, when including displacements, postselection allows to emulate more complex inputs states \cite{hamilton2024boson}.
Finally, a natural extension of our work is excluding postselection on collision-free events and generalizing the parity observables to other operators, e.g., as introduced for Gaussian Bosonic Born machines in \cite{kolarovszki2026learning}.

\smallskip

\textit{Note added.} Ref.\ \cite{kolarovszki2026learning} was put out in parallel to this work. During the preparation of the manuscript, we also became aware of \cite{gottlieb2026efficient}.

\section*{Acknowledgements}
UC~thanks H.\ Ollivier and H.\ Thomas for interesting discussions. VD, and AK acknowledge the support from the Dutch National Growth Fund (NGF), as part of the Quantum Delta NL programme. VD acknowledges support from the Dutch Research Council (NWO/OCW), as part of the Quantum Software Consortium programme (project number 024.003.03). This project was also co-funded by the European Union (ERC CoG, BeMAIQuantum, 101124342). ZK and BB would like to thank the support of the Hungarian National Research, Development and Innovation Office (NKFIH) through the KDP-2023 funding scheme, the Quantum Information National Laboratory of Hungary and the grants TKP-2021-NVA-04, TKP2021-NVA-29 and FK 135220. ZZ acknowledges funding from the Research Council of Finland through the Finnish Quantum Flagship project 358878.

\appendix
\section{Estimation of parity word expectation values}\label{app:algorithm}
    In this appendix we give an explicit procedure for estimating parity word expectation values in BSBMs.

    Let us consider a subset of modes defined by $\alpha \in \{ 0, 1 \}^m$ and denote the corresponding parity word supported on the modes given by $\alpha$ as $\Pi_\alpha \coloneqq \bigotimes_{j = 1}^m \Pi_j^{\alpha_j}$.
    Moreover, the input state can be written as
    \begin{equation}
        \ket{\psi_{\text{in}}} = \ket{1^k 0^{m-k}},
    \end{equation}
    to which we apply an interferometer $\hat{U}$ modeled by a unitary matrix $U \in \mathcal{U}(m)$.

    To obtain the expectation values of the parity words for $\hat{U}\ket{\psi_{\text{in}}}$, we can utilize Gurvits' algorithm~\cite{gurvits2005complexity} for approximating the permanent. We can express a parity word expectation value as
    \begin{equation}
         \bra{\psi_{\text{in}}} \hat{U}^\dagger \Pi_\alpha \hat{U} \ket{\psi_{\text{in}}} = \bra{\psi_{\text{in}}} \hat{V} \ket{\psi_{\text{in}}} = \operatorname{Per}(V_{[k], [k]}),
    \end{equation}
    where we denoted $\hat{V} \coloneqq \hat{U}^\dagger \Pi_\alpha \hat{U}$, to which we associate a single-particle unitary matrix $V = U^\dagger \operatorname{diag}(1-2\alpha) U$ with $\operatorname{diag}$ denoting the operation of forming a diagonal matrix from the specified vector. Moreover, $V_{[k],[k]}$ denotes the principal $k\times k$ submatrix indexed by the occupied input modes.

    Gurvits' estimator is based on the identity
    \begin{equation}
        \operatorname{Per}(W) = \underset{x \in \{-1, 1\}^k }{\mathbb{E}} [\mathrm{Rys}_{x}(W)],
    \end{equation}
    where we denoted $W \coloneqq V_{[k], [k]}$ for brevity and
    \begin{equation}
        \mathrm{Rys}_{x}(W) \coloneqq \left( \prod_{i=1}^k x_i \right) \left( \prod_{i=1}^k \sum_{j=1}^k W_{ij} x_j \right),
    \end{equation}
    so that $\mathrm{Rys}_x(W)$ is an unbiased estimator of $\operatorname{Per}(W)$.
    Each sample $\mathrm{Rys}_x(W)$ can be computed in $\mathcal O(k^2)$ time. In our setting, $\lvert \mathrm{Rys}_x(W)\rvert\le 1$, hence averaging $T=\mathcal O(\log(\delta^{-1})/\varepsilon^2)$ independent samples and applying a Hoeffding bound yields an $\varepsilon$-additive estimate with success probability at least $1-\delta$. The resulting overall cost is $\mathcal O\!\big(k^2 \log(\delta^{-1})/\varepsilon^2\big)$. We summarize the full workflow in \Cref{alg:parity_estimation}.

    \begin{algorithm}[H]
    \caption{Estimation of parity word expectation values}
        \label{alg:parity_estimation}
        \begin{algorithmic}[1]
        \REQUIRE Unitary $U \in \mathcal{U}(m)$, parity subset vector $\alpha \in \{0, 1\}^m$, number of input photons $k$, number of samples $N$.
        \ENSURE Estimate $\hat{E}$ of the expectation value $\langle \psi_{\text{in}} | \hat{U}^\dagger \Pi_\alpha \hat{U} | \psi_{\text{in}} \rangle$.
        
        \STATE \textbf{Pre-processing:}
        \STATE Construct diagonal matrix $D \leftarrow \operatorname{diag}(1 - 2\alpha)$.
        \STATE Compute $V = U^\dagger D U$.
        \STATE Extract the upper-left $k \times k$ submatrix: $W \leftarrow V_{[k], [k]}$.
        
        \STATE \textbf{Monte Carlo estimation:}
        \STATE Initialize $S \leftarrow 0$.
        \FOR{$n = 1$ \TO $N$}
            \STATE Sample $x \in \{-1, 1\}^k$ uniformly at random.
            \STATE Compute the product of signs: $P_x \leftarrow \prod_{i=1}^k x_i$.
            \STATE Compute the row-sum product: $R_x \leftarrow \prod_{i=1}^k \left( \sum_{j=1}^k W_{ij} x_j \right)$.
            \STATE $S \leftarrow S + P_x \cdot R_x$.
        \ENDFOR
        
        \STATE \textbf{Final Result:}
        \STATE $\hat{E} \leftarrow S / N$.
        \RETURN $\hat{E}$
        \end{algorithmic}
    \end{algorithm}

\bibliographystyle{plain}
\bibliography{refs}

\end{document}